\documentclass[preprint,12pt]{elsarticle}

\usepackage{amssymb}
\usepackage{amsmath}
\usepackage{amsthm}
\usepackage{tabularx}
\usepackage{array}
\usepackage{multirow}
\usepackage{booktabs}
\usepackage{empheq}
\usepackage{graphicx}
\usepackage{xcolor}
\usepackage{hyperref}
\usepackage{lineno}

\newtheorem{theorem}{Theorem}
\newtheorem{lemma}{Lemma}

\theoremstyle{remark}
\newtheorem{remark}{Remark}

\newcommand{\D}{\mathbf{D}}

\newcommand{\U}{\mathbf{U}}
\newcommand{\F}{\mathbf{F}}
\newcommand{\G}{\mathbf{G}}

\newcommand{\W}{\mathbf{W}}

\newcommand{\Fone}{\mathbf{F}^{\mathrm{1D}}}

\begin{document}

\begin{frontmatter}

\title{Towards Exact Multidimensional Fluxes: A Structure-Preserving Data-Informed Framework for Riemann Solvers}

\author[label1]{Yucheng Zhang\corref{cor1}}
\author[label1,label3]{Clint Dawson}

\cortext[cor1]{Corresponding author. Email: yucheng@utexas.edu}

\affiliation[label1]{organization={Oden Institute for Computational Engineering and Sciences, University of Texas at Austin},
            city={Austin},
            state={Texas},
            country={USA}}


\affiliation[label3]{organization={Department of Aerospace Engineering and Engineering Mechanics, University of Texas at Austin},
            city={Austin},
            state={Texas},
            country={USA}}

\begin{abstract}

\! Finite volume and discontinuous Galerkin methods evolve the discrete solution using space--time flux integrals over mesh faces, which typically depend on solutions to local Riemann problems. In the multidimensional Riemann problem setting, this integral is difficult to evaluate exactly because of complex wave interactions. In this work, we show that the exact space--time flux integral is a polynomial in the time step and face dimensions. This motivates a learning-based approach, as only the mapping from the initial states to the polynomial coefficients needs to be learned. We propose a structure-preserving, data-informed Riemann solver framework that uses a neural network trained on high-resolution simulation data to learn this mapping. Key structural properties, including consistency and mirror symmetry, are enforced by construction. Provided that the high-resolution simulations converge to the exact flux integrals and the network accurately predicts the coefficients, this framework yields near-exact flux integrals for use in numerical simulations.

We demonstrate the framework using classical implosion and explosion problems for the two-dimensional Euler equations. In an unsplit MUSCL--Hancock scheme, the learned Riemann solver remains stable at a CFL number of 0.9, whereas the same scheme using the exact one-dimensional Riemann solver fails. Compared with the latter scheme at CFL=0.45, the learned scheme is slightly more accurate and requires approximately 30\% less computational time.

\end{abstract}


\end{frontmatter}

\section{Introduction}
\label{sec:intro}

Godunov-type methods compute the flux across each cell face by solving a local Riemann problem. In one spatial dimension, the time-averaged flux follows directly from the self-similar solution to a single Riemann problem. In multiple spatial dimensions, Riemann solvers can remain one-dimensional and are applied along the direction normal to each face at the quadrature points. When used without additional transverse coupling, these face-normal Riemann solvers do not directly account for transverse wave interactions, leading to mesh imprinting and a stable time step that decreases with the number of spatial dimensions.

Multidimensional Riemann solvers address these issues. For linear acoustics, the multidimensional Riemann problem can be solved exactly~\cite{barsukow2022exact}. For nonlinear systems such as the Euler equations, however, the two-dimensional Riemann problem has no closed-form solution and is known only numerically~\cite{schulzrinne1993numerical,lax1998solution}. In a series of pioneering papers, Balsara and co-workers extended the HLL and HLLC Riemann solvers to multiple space dimensions on both structured and unstructured meshes~\cite{balsara2010multidimensional,balsara2012two,balsara2014multidimensional,balsara2015multidimensional}. As in one dimension, these solvers represent the wave pattern with simplified wave models bounded by estimated wave speeds. They have become a key component of divergence-free schemes for magnetohydrodynamics~\cite{balsara2015divergence}. More recently, Gaburro et al.~\cite{gaburro2025general} introduced two complete multidimensional Riemann solvers for unstructured polygonal meshes: a multidimensional extension of the Osher--Solomon solver and an upwind solver based on residual distribution schemes. Other approaches to multidimensional Riemann solvers can be found in~\cite{ROE1986458, Colella1990, Saltzman1994, wendroff1999two,brio2001two,vides2015simple}. For nonlinear systems, all of these solvers approximate the multidimensional wave interaction with a simplified model. A flux integral close to the exact one would require a high-resolution numerical solution of the multidimensional Riemann problem at every cell vertex and at every time step, which is far too expensive for a simulation.

In recent years, machine learning has shown considerable promise in computational fluid dynamics (CFD)~\cite{brunton2020machine,vinuesa2022enhancing}. Popular approaches include neural time steppers that learn the update between time steps~\cite{pfaff2021learning}, physics-informed neural networks (PINNs)~\cite{raissi2019physics}, and neural operators such as DeepONet~\cite{lu2021learning} and the Fourier neural operator~\cite{li2021fourier}. For hyperbolic conservation laws, neural networks have been integrated into classical shock capturing schemes as troubled cell indicators~\cite{ray2018artificial}, artificial viscosity models~\cite{discacciati2020controlling}, and nonlinear reconstructions~\cite{stevens2020enhancement,wang2020learning,kossaczka2021enhanced,bezgin2021data,bezgin2022weno3}. A closely related line of work learns the numerical flux of a conservative finite volume scheme, either to improve accuracy or to infer unknown conservation laws from trajectory data, while preserving discrete conservation by construction~\cite{chen2024learning,morand2024deep,kim2025approximating,suh2025tvd,liu2026entropy,liu2026neural,baba2026supervised}. Learning-based Riemann solvers have also been developed~\cite{magiera2020constraint,ruggeri2022neural,gyrya2024machine,peyvan2024riemannonets,wu5295012rimnet}. In our previous work~\cite{zhang2026learning}, we proposed a hard-constrained neural Riemann solver (HCNRS) that enforces crucial constraints, such as consistency, mirror symmetry, and scaling invariance. We showed that when HCNRS is integrated into the numerical schemes, it reproduces the exact Riemann solver solutions with high accuracy. To our knowledge, a data-driven approach to multidimensional Riemann solvers has not been studied. One obstacle is identifying a suitable learning target: the space--time flux integral appears to have a complicated dependence on the initial states, time step, and face dimension.

In this work, we show that the exact space--time flux integral in multidimensional Riemann problem settings is a polynomial in the time step and face dimension, provided that the waves from the vertex stay within the face during the time step, which holds for CFL numbers up to one. This reduces the learning task to predicting polynomial coefficients from the initial states, with no dependence on the time step or face size. Building on this result, we develop a neural Riemann solver for the two-dimensional Euler equations that enforces consistency, mirror symmetry, scaling invariance, and rotational invariance. The network is trained on high-resolution simulation data, computed once and offline. We test the solver in an unsplit MUSCL--Hancock scheme on the classical implosion and explosion problems.

The remainder of the paper is organized as follows. Section~\ref{sec:faceflux} derives the polynomial representation of the flux integral. Section~\ref{sec:network} presents the neural network architecture. Section~\ref{sec:results} describes the training setup, the numerical method and results, and Section~\ref{sec:conclusions} concludes the paper.


\section{Polynomial dependence of the space--time flux integral on the time step and face dimensions}
\label{sec:faceflux}

We consider isolated local Riemann problems with piecewise constant initial data. We assume the following properties:
\begin{itemize}
\item[(S)] \emph{Self-similarity.}
The solution has the form
$\mathbf U(\mathbf x,t)=\tilde{\mathbf U}(\mathbf x/t)$.

\item[(F)] \emph{Finite domain of dependence.}
There is a common speed bound $\lambda>0$ such that solutions agree at $(\mathbf x,t)$ whenever their initial data agree within distance $\lambda t$ of $\mathbf x$.
\end{itemize}
We also assume that the normal fluxes used below are well defined and integrable on every bounded subset of their similarity domains.

The restrictions on $\lambda\Delta t$ below ensure that the localized flux corrections remain within the integration region. On a Cartesian face, these restrictions can be expressed in terms of the directional Courant numbers $\nu_j=\lambda\Delta t/\Delta x_j$, requiring
$\nu_j\le 1$ in each tangential direction.

We first present detailed derivations in one, two, and three space dimensions, and then extend the result to an arbitrary number of space dimensions using an inclusion--exclusion decomposition of the face flux.

\subsection{One space dimension}

Let $\mathbf F^{\rm 1D}=\mathrm{ExactRS}(\mathbf U_L,\mathbf U_R)$ be the flux of the exact Riemann solver at an interface $x=0$ with left and right states $\mathbf U_L$ and $\mathbf U_R$. By (S), $\mathbf U(0,t)=\tilde{\mathbf U}(0)$ does not change in time, so the space--time flux integral is
\[
\mathbf I=\int_0^{\Delta t}\mathbf F\big(\mathbf U(0,t)\big)\,dt=\mathbf F^{\rm 1D}\,\Delta t ,
\]
a polynomial of degree one in $\Delta t$. The face is a point, so there is no spatial degree.

\subsection{Two space dimensions}

\begin{figure}[t]
\centering
\includegraphics[width=\textwidth]{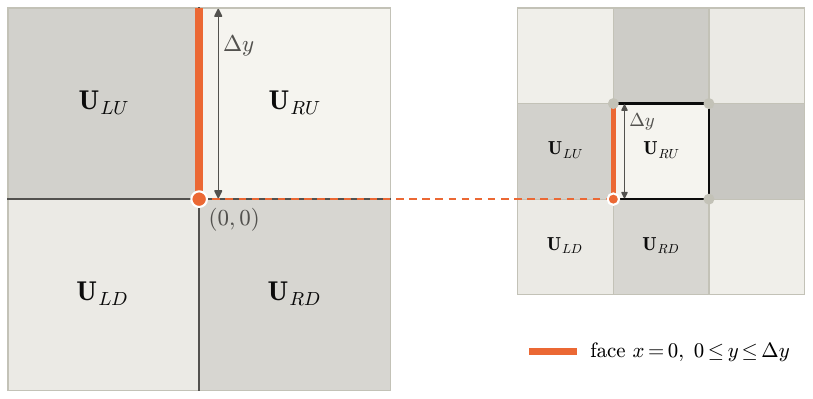}
\caption{Left: the 2D Riemann problem at a vertex at the origin, with constant states $\mathbf U_{RU},\mathbf U_{LU},\mathbf U_{LD},\mathbf U_{RD}$ in the four quadrants. Right: a structured mesh, on which a finite volume scheme updates the white cell. The left face $x=0$, $0\le y\le\Delta y$ (orange) of this cell lies on the upward half-axis of the Riemann problem at its lower vertex (dashed line). Near this vertex, the flux through the face is given by the Riemann problem of the left panel. The other vertices of the cell (grey dots) have 2D Riemann problems of their own.}

\label{fig:2d}
\end{figure}

Consider a vertex at the origin, where four cells meet with constant states $\mathbf U_{LU},\mathbf U_{RU},\mathbf U_{LD},\mathbf U_{RD}$ in the four quadrants (left/right, up/down). Take the face $x=0$, $0\le y\le\Delta y$, and let $\mathbf F$ be the flux in the $x$-direction. We consider the space-time flux integral $\mathbf I$ on this face with time step $ \Delta t$.

Figure~\ref{fig:2d} shows where this setting arises. A finite volume
scheme on a structured mesh updates the white cell using the flux
through each face, integrated over the face and over the time step.
The four cells meeting at the lower endpoint of the left face supply
the four constant states of the local Riemann problem. Computing the flux through the mesh face in the finite volume scheme also requires accounting for the Riemann problem
at its upper endpoint (Section \ref{sec:scheme}).

By (S), define
\begin{equation}\label{eq:f2d}
\mathbf f(\eta) \equiv \mathbf F\big(\mathbf U(x=0,y,t)\big),\qquad \eta=\frac yt .
\end{equation}
By (F), a point of the face with $y>\lambda t$ does not receive information from the cells below the vertex. It sees only the one-dimensional Riemann problem between $\mathbf U_{LU}$ and $\mathbf U_{RU}$, so
\begin{equation}\label{eq:far2d}
\mathbf f(\eta)=\mathbf F^{\rm 1D}=\mathrm{ExactRS}(\mathbf U_{LU},\mathbf U_{RU})\quad\text{for }\eta>\lambda .
\end{equation}
We define the \emph{integrated flux difference}
\begin{equation}\label{eq:D}
\mathbf D=\int_0^\infty\big(\mathbf f(\eta)-\mathbf F^{\rm 1D}\big)\,d\eta .
\end{equation}
By \eqref{eq:far2d}, the integrand vanishes for $\eta>\lambda$, so $\mathbf D$ is finite. The integrated flux difference depends only on the four initial states. It does not depend on $t$, $\Delta t$, $\Delta x$ or $\Delta y$.

\begin{lemma}[Two dimensions]\label{lem:2d}
If $\lambda\Delta t\le\Delta y$, then for every $0<t\le\Delta t$
\begin{equation}\label{eq:2d-t}
\int_0^{\Delta y}\mathbf F\big(\mathbf U(0,y,t)\big)\,dy=\Delta y\,\mathbf F^{\rm 1D}+t\,\mathbf D ,
\end{equation}
and the space--time flux integral is
\begin{equation}\label{eq:2d-I}
\mathbf I=\int_0^{\Delta t}\!\!\int_0^{\Delta y}\mathbf F\big(\mathbf U(0,y,t)\big)\,dy\,dt=\Delta y\,\Delta t\,\mathbf F^{\rm 1D}+\frac{\Delta t^2}{2}\,\mathbf D 
\end{equation}
\end{lemma}

\begin{proof}
With \eqref{eq:f2d},
\[
\int_0^{\Delta y}\mathbf F\big(\mathbf U(0,y,t)\big)\,dy=t\int_0^{\Delta y/t}\mathbf f(\eta)\,d\eta .
\]
Writing $\mathbf f=\mathbf F^{\rm 1D}+(\mathbf f-\mathbf F^{\rm 1D})$,
\[
t\int_0^{\Delta y/t}\mathbf f(\eta)\,d\eta=\Delta y\,\mathbf F^{\rm 1D}+t\int_0^{\Delta y/t}\big(\mathbf f(\eta)-\mathbf F^{\rm 1D}\big)\,d\eta .
\]
Since $t\le\Delta t\le\Delta y/\lambda$, the upper limit satisfies $\Delta y/t\ge\lambda$. The integrand vanishes beyond $\lambda$ by \eqref{eq:far2d}, hence:
\[
\int_0^{\Delta y/t}\big(\mathbf f(\eta)-\mathbf F^{\rm 1D}\big)\,d\eta=\int_0^{\infty}\big(\mathbf f(\eta)-\mathbf F^{\rm 1D}\big)\,d\eta=\mathbf D .
\]
This proves \eqref{eq:2d-t}. Integrating it over the time step proves \eqref{eq:2d-I}
\end{proof}

\begin{remark}[Unstructured meshes]\label{rem:unstructured}
The proof of Lemma~\ref{lem:2d} uses only (S) and (F), so it also applies at a vertex of a general polygonal mesh, with a constant state in each sector between two edges. Along a straight face of length $\Delta s$ issuing from the vertex, the normal flux equals that of the one-dimensional Riemann problem between the two adjacent states beyond a distance proportional to $\lambda t$. For a sufficiently small time step, therefore, $\mathbf I=\Delta s\,\Delta t\,\mathbf F^{\rm 1D}+\frac{\Delta t^2}{2}\,\mathbf D$, where $\mathbf D$ is defined by \eqref{eq:D} with the normal flux and depends on the initial states and the angles between the edges.
\end{remark}

\subsection{Three space dimensions}

Consider a vertex $v$ at the origin, where eight cells meet with constant states $\mathbf U_{\pm\pm\pm}$ in the eight octants, the three signs being those of $x$, $y$ and $z$. Take the face $x=0$, $0\le y\le\Delta y$, $0\le z\le\Delta z$, and let $\mathbf F$ be the flux in the $x$-direction. 

Figure~\ref{fig:3d} shows this setting. What is new in three dimensions is that two sides of the face lie on mesh edges through $v$, namely $e_1=\{x=z=0,\ y>0\}$ and $e_2=\{x=y=0,\ z>0\}$, and four cells meet around each of these edges.

\begin{figure}[t]
\centering
\includegraphics[width=\textwidth]{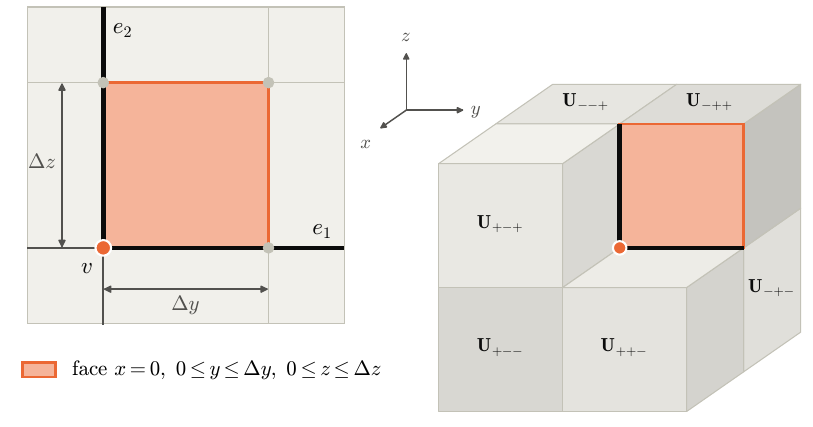}
\caption{Left: the face $x=0$, $0\le y\le\Delta y$, $0\le z\le\Delta z$ (orange), seen from $x>0$. Its lower-left corner is the vertex $v$ at the origin, and its two sides through $v$ lie on the mesh edges $e_1$ and $e_2$ (thick lines). The thin lines are other mesh edges in the plane $x=0$, and the grey dots are the other vertices of the face. Right: the Riemann problem at $v$, with constant states $\mathbf U_{\pm\pm\pm}$ in the eight octants; the signs are those of $x$, $y$ and $z$. }
\label{fig:3d}

\end{figure}

By (S), define
\begin{equation}\label{eq:f3d}
\mathbf f_v(\eta,\zeta) \equiv \mathbf F\big(\mathbf U(0,y,z,t)\big), \qquad \eta=y/t,  \qquad \zeta=z/t
\end{equation}

By (F), a point of the face with $y>\lambda t$ does not feel the four cells with $y<0$. It sees only the \emph{edge problem} of $e_1$, and its flux on the face is a function $\mathbf f_{e_1}(\zeta)$ of $\zeta$ alone. Likewise, a point with $z>\lambda t$ sees only the edge problem of $e_2$, whose flux on the face is a function $\mathbf f_{e_2}(\eta)$. As in \eqref{eq:far2d}, a point of an edge problem farther than $\lambda t$ from its edge sees only the one-dimensional problem between $\mathbf U_{-++}$ and $\mathbf U_{+++}$. Hence
\begin{equation}\label{eq:far3d}
\begin{aligned}
&\mathbf f_v(\eta,\zeta)=\mathbf f_{e_1}(\zeta)\ \text{ for }\eta>\lambda, &\qquad&\mathbf f_{e_1}(\zeta)=\mathbf F^{\rm 1D}\ \text{ for }\zeta>\lambda,\\
&\mathbf f_v(\eta,\zeta)=\mathbf f_{e_2}(\eta)\ \text{ for }\zeta>\lambda, &&\mathbf f_{e_2}(\eta)=\mathbf F^{\rm 1D}\ \text{ for }\eta>\lambda,
\end{aligned}
\end{equation}
where $\mathbf F^{\rm 1D}=\mathrm{ExactRS}(\mathbf U_{-++},\mathbf U_{+++})$.

We define the integrated flux differences of the two edges and of the vertex,
\begin{equation}\label{eq:D3d}
\begin{aligned}
\mathbf D_{e_1}&=\int_0^\infty\big(\mathbf f_{e_1}(\zeta)-\mathbf F^{\rm 1D}\big)\,d\zeta,\qquad
\mathbf D_{e_2}=\int_0^\infty\big(\mathbf f_{e_2}(\eta)-\mathbf F^{\rm 1D}\big)\,d\eta,\\
\mathbf D_v&=\int_0^\infty\!\!\int_0^\infty\big(\mathbf f_v(\eta,\zeta)-\mathbf f_{e_1}(\zeta)-\mathbf f_{e_2}(\eta)+\mathbf F^{\rm 1D}\big)\,d\eta\,d\zeta .
\end{aligned}
\end{equation}
$\mathbf D_{e_1}$ and $\mathbf D_{e_2}$ are the integrated flux differences for the two edge problems, as in \eqref{eq:D}. $\mathbf D_v$ is new in three dimensions. By \eqref{eq:far3d}, the integrands of $\mathbf D_{e_1}$ and $\mathbf D_{e_2}$ vanish for $\zeta>\lambda$ and $\eta>\lambda$, respectively. If $\eta>\lambda$, then $\mathbf f_v=\mathbf f_{e_1}$ and $\mathbf f_{e_2}=\mathbf F^{\rm 1D}$, so the integrand of $\mathbf D_v$ vanishes. Likewise, it vanishes for $\zeta>\lambda$. So $\mathbf D_v$ is an integral over the square $0\le\eta,\zeta\le\lambda$, and all three are finite.

\begin{lemma}[Three dimensions]\label{lem:3d}
If $\lambda\Delta t\le\min(\Delta y,\Delta z)$, then for every $0<t\le\Delta t$
\begin{equation}\label{eq:3d-t}
\begin{split}
&\int_0^{\Delta z}\!\!\int_0^{\Delta y}\mathbf F\big(\mathbf U(0,y,z,t)\big)\,dy\,dz\\
&\qquad=\Delta y\,\Delta z\,\mathbf F^{\rm 1D}+t\,\big(\Delta y\,\mathbf D_{e_1}+\Delta z\,\mathbf D_{e_2}\big)+t^2\,\mathbf D_v ,
\end{split}
\end{equation}
and the space--time flux integral is
\begin{equation}\label{eq:3d-I}
\begin{split}
\mathbf I&=\int_0^{\Delta t}\!\!\int_0^{\Delta z}\!\!\int_0^{\Delta y}\mathbf F\big(\mathbf U(0,y,z,t)\big)\,dy\,dz\,dt\\
&=\Delta y\,\Delta z\,\Delta t\,\mathbf F^{\rm 1D}+\frac{\Delta t^2}{2}\big(\Delta y\,\mathbf D_{e_1}+\Delta z\,\mathbf D_{e_2}\big)+\frac{\Delta t^3}{3}\,\mathbf D_v .
\end{split}
\end{equation}
\end{lemma}

\begin{proof}
Let $R_t=[0,\Delta y/t]\times[0,\Delta z/t]$. With \eqref{eq:f3d},
\[
\int_0^{\Delta z}\!\!\int_0^{\Delta y}\mathbf F\big(\mathbf U(0,y,z,t)\big)\,dy\,dz=t^2\iint_{R_t}\mathbf f_v(\eta,\zeta)\,d\eta\,d\zeta .
\]
Since $t\le\Delta t\le\min(\Delta y,\Delta z)/\lambda$, both sides of $R_t$ are at least $\lambda$, so $R_t$ contains the square $[0,\lambda]^2$. Write $\mathbf f_v$ as the sum of four parts,
\[
\mathbf f_v=\mathbf F^{\rm 1D}+\big(\mathbf f_{e_1}-\mathbf F^{\rm 1D}\big)+\big(\mathbf f_{e_2}-\mathbf F^{\rm 1D}\big)+\big(\mathbf f_v-\mathbf f_{e_1}-\mathbf f_{e_2}+\mathbf F^{\rm 1D}\big),
\]
and integrate them over $R_t$ one at a time.
\begin{itemize}
\item The constant $\mathbf F^{\rm 1D}$ gives $\Delta y\,\Delta z\,\mathbf F^{\rm 1D}$.
\item The second part depends only on $\zeta$, so the $\eta$-integral gives the factor $\Delta y/t$. The upper limit $\Delta z/t\ge\lambda$ of the $\zeta$-integral can be replaced by $\infty$, because the integrand vanishes beyond $\lambda$:
\[
t^2\,\frac{\Delta y}{t}\int_0^{\Delta z/t}\big(\mathbf f_{e_1}(\zeta)-\mathbf F^{\rm 1D}\big)\,d\zeta=t\,\Delta y\,\mathbf D_{e_1}.
\]
\item Likewise, the third part gives $t\,\Delta z\,\mathbf D_{e_2}$.
\item The fourth part vanishes outside the square $[0,\lambda]^2\subset R_t$, so $R_t$ can be replaced by the whole quarter plane, and it gives $t^2\,\mathbf D_v$.
\end{itemize}
Adding the four parts proves \eqref{eq:3d-t}. Integrating it over the time step proves \eqref{eq:3d-I}
\end{proof}

\subsection{Any number of space dimensions}
For $d\ge 2$, let $\mathbf F$ be the flux in the $x_1$-direction of $\mathbf x=(x_1,\dots,x_d)$, and let $\mathcal T=\{2,\dots,d\}$ be the set of tangential directions.
Consider an isolated Riemann problem with one constant state in each of the $2^d$ orthants meeting at the origin, and the face
\[
P=\{\mathbf x:\ x_1=0,\ 0\le x_j\le\Delta x_j
\ \text{for }j\in\mathcal T\},
\]
which has the origin $v$ as a corner. For $S\subseteq\mathcal T$, write $\eta_S=(x_j/t)_{j\in S}$.

For each $S\subseteq\mathcal T$, let $P_S$ be the part of $P$ where $x_j=0$ for all $j\in S$. Its measure in its own dimension is
\[
|P_S|=\prod_{j\in\mathcal T\setminus S}\Delta x_j,
\]
with $|P_{\mathcal T}|=1$. Thus $P_\emptyset=P$ is the face itself and $P_{\mathcal T}=\{v\}$ is the vertex.

To define the local Riemann problem associated with $P_S$, take the $2^{|S|+1}$ states adjacent to its relative interior and extend them uniformly in the directions $\mathcal T\setminus S$. These states have positive signs in all directions $j\in\mathcal T\setminus S$. By (S), the resulting flux on $x_1=0$ is a function $\mathbf f_S(\eta_S)$. In particular, $\mathbf f_\emptyset=\mathbf F^{\rm 1D}$, and $\mathbf f_{\mathcal T}$ is the flux of the vertex problem.

By (F), a point with $x_j>\lambda t$ does not feel the cells with $x_j<0$. Removing these cells from the problem of $P_S$ gives the problem of $P_{S\setminus\{j\}}$, so
\begin{equation}\label{eq:fard}
\mathbf f_S(\eta_S)=\mathbf f_{S\setminus\{j\}}(\eta_{S\setminus\{j\}})\qquad\text{if } j\in S\text{ and }\eta_j>\lambda .
\end{equation}
This contains \eqref{eq:far2d} and \eqref{eq:far3d}.

We define the integrated flux difference of $P_S$ by inclusion--exclusion,
\begin{equation}\label{eq:Dd}
\begin{aligned}
\mathbf D_S
&=\int_{[0,\infty)^{|S|}}
\boldsymbol\psi_S(\eta_S)\,d\eta_S,\\
\boldsymbol\psi_S(\eta_S)
&=\sum_{R\subseteq S}
(-1)^{|S|-|R|}\,\mathbf f_R(\eta_R),
\end{aligned}
\end{equation}
with $\mathbf D_\emptyset=\mathbf F^{\rm 1D}$.
For nonempty $S$, fix $j\in S$ and suppose $\eta_j>\lambda$. The terms in $\boldsymbol\psi_S$ then cancel in pairs: for each $R\subseteq S\setminus\{j\}$, the terms indexed by $R$ and $R\cup\{j\}$ have opposite signs and equal flux values by \eqref{eq:fard}. Hence $\boldsymbol\psi_S$ vanishes outside $[0,\lambda]^{|S|}$. Together with the assumed integrability of the flux profiles, this ensures that $\mathbf D_S$ is finite.

\begin{theorem}[Any number of space dimensions]\label{thm:nd}
If $\lambda\Delta t\le\min_{j\in\mathcal T}\Delta x_j$, then for every $0<t\le\Delta t$
\begin{equation}\label{eq:nd-t}
\int_P\mathbf F\big(\mathbf U(\mathbf x,t)\big)\,dA=\sum_{S\subseteq\mathcal T}t^{|S|}\,|P_S|\,\mathbf D_S ,
\end{equation}
where $dA=dx_2\cdots dx_d$, and the space--time flux integral is
\begin{equation}\label{eq:nd-I}
\mathbf I=\int_0^{\Delta t}\!\!\int_P\mathbf F\big(\mathbf U(\mathbf x,t)\big)\,dA\,dt=\sum_{S\subseteq\mathcal T}\frac{\Delta t^{|S|+1}}{|S|+1}\,|P_S|\,\mathbf D_S .
\end{equation}
\end{theorem}

\begin{proof}
By (S), $\mathbf F(\mathbf U(\mathbf x,t))=\mathbf f_{\mathcal T}(\eta_{\mathcal T})$ on the face. 
By inclusion--exclusion,
\[
\begin{aligned}
\sum_{S\subseteq\mathcal T}\boldsymbol\psi_S(\eta_S)
&=
\sum_{R\subseteq\mathcal T}\mathbf f_R(\eta_R)
\sum_{R\subseteq S\subseteq\mathcal T}
(-1)^{|S|-|R|}\\
&=\mathbf f_{\mathcal T}(\eta_{\mathcal T}),
\end{aligned}
\]
since the inner sum is zero unless $R=\mathcal T$, in which
case it equals one.

We integrate each part over $P$ at a fixed time $t$. The part $\boldsymbol\psi_S$ depends only on $x_j$ with $j\in S$, so the integrals over the other directions give the factor $|P_S|$. In the remaining directions, substitute $x_j=t\eta_j$, so that $\prod_{j\in S}dx_j=t^{|S|}\,d\eta_S$ and the domain of integration becomes the box $\prod_{j\in S}[0,\Delta x_j/t]$. Since $t\le\Delta t\le\Delta x_j/\lambda$, this box contains the cube $[0,\lambda]^{|S|}$, outside which $\boldsymbol\psi_S$ vanishes. The box can therefore be replaced by $[0,\infty)^{|S|}$, and
\[
\int_P\boldsymbol\psi_S(\eta_S)\,dA=t^{|S|}\,|P_S|\,\mathbf D_S .
\]
Summing over $S$ proves \eqref{eq:nd-t}. Integrating over the time step, with $\int_0^{\Delta t}t^m\,dt=\Delta t^{m+1}/(m+1)$, gives \eqref{eq:nd-I}.
\end{proof}

For fixed initial states, \eqref{eq:nd-I} is therefore a polynomial
of degree at most $d$ in $\Delta t$ and of degree at most one in
each tangential side length $\Delta x_j$, within the stated
time-step restrictions. Every nonzero monomial has total degree
$d$ in the time step and the tangential side lengths.


\section{Neural network architecture}
\label{sec:network}

In this work, we focus on the two-dimensional Euler equations for an ideal gas,
\begin{equation}
\partial_t\U+\partial_x\F(\U)+\partial_y\G(\U)=\mathbf 0,
\label{eq:euler}
\end{equation}
with
\begin{equation}
\U=\begin{pmatrix}\rho\\ \rho u\\ \rho v\\ E\end{pmatrix},\qquad
\F(\U)=\begin{pmatrix}\rho u\\ \rho u^2+p\\ \rho uv\\ u(E+p)\end{pmatrix},\qquad
\G(\U)=\begin{pmatrix}\rho v\\ \rho uv\\ \rho v^2+p\\ v(E+p)\end{pmatrix},
\label{eq:euler_flux}
\end{equation}
and the equation of state $p=(\gamma-1)\big(E-\tfrac12\rho\,(u^2+v^2)\big)$ with $\gamma=1.4$. Here $\rho$ is the density, $u$ and $v$ are the velocity components in the $x$- and $y$-directions, $E$ is the total energy and $p$ is the pressure. 

For the 2-dimensional Riemann solver, we consider the four properties below:
\begin{enumerate}
  \item \textbf{Consistency (reduction to 1D):} if the data are one-dimensional across the face, the flux through the face is the one-dimensional Riemann flux $\Fone$ everywhere, so $\D$ must vanish. This includes four coinciding states.
  \item \textbf{Mirror symmetry:} reflecting the vertex problem across the line of the face must reflect $\D$, so that its mass, tangential-momentum and energy components change sign while its normal-momentum component does not.
  \item \textbf{Scaling invariance:} under certain scalings of the state variables, the wave pattern remains unchanged, and $\D$ must scale accordingly.
  \item \textbf{Rotational invariance:} a rotation of the vertex problem by $90^\circ$ maps its four faces onto each other, so the integrated flux differences of the four faces must follow from one another by rotation.
\end{enumerate}

We now describe how the properties are enforced.
\subsection{Rotational invariance}

Rotation is enforced by learning the integrated flux difference $\D$ of the face above a vertex, Eq.~\eqref{eq:D}. The integrated flux differences of the other three faces are obtained by applying the same network to the vertex problem rotated so that the face points up, and rotating the result back. In the remainder of this section, we describe the network for the face above the vertex.

\subsection{Scaling invariance}
\label{sec:net_scaling}
Similar to~\cite{zhang2026learning}, scaling invariance is enforced through an invariant-based reduction of the inputs. Assume
\begin{equation}
\mathcal{N\!N}(\W_{RU},\W_{LU},\W_{LD},\W_{RD})=\D.
\end{equation}
The scaling laws imply that, for any $\alpha>0$ and $\beta>0$,
\begin{equation}
\mathcal{N\!N}\big(\alpha\rho_q,\ \beta u_q,\ \beta v_q,\ \alpha\beta^2p_q\big)=\alpha\,\mathrm{diag}\big(\beta^2,\beta^3,\beta^3,\beta^4\big)\,\D,
\label{eq:nn_scaling}
\end{equation}
where $q\in\{RU,LU,LD,RD\}$ label the four quadrants

Define the reference density, pressure and speed
\begin{equation}
\rho_{\mathrm{ref}}:=\frac14\sum_q\rho_q,\qquad p_{\mathrm{ref}}:=\frac14\sum_q p_q,\qquad c_{\mathrm{ref}}:=\sqrt{p_{\mathrm{ref}}/\rho_{\mathrm{ref}}},
\label{eq:refs}
\end{equation}
and choose $\alpha=1/\rho_{\mathrm{ref}}$ and $\beta=1/c_{\mathrm{ref}}$. Since $\rho_{\mathrm{ref}}c_{\mathrm{ref}}^2=p_{\mathrm{ref}}$, the normalized states
\begin{equation}
\hat\W_q:=\Big(\frac{\rho_q}{\rho_{\mathrm{ref}}},\ \frac{u_q}{c_{\mathrm{ref}}},\ \frac{v_q}{c_{\mathrm{ref}}},\ \frac{p_q}{p_{\mathrm{ref}}}\Big)
\label{eq:Whats}
\end{equation}
satisfy
\begin{equation}
\mathcal{N\!N}\big(\hat\W_{RU},\hat\W_{LU},\hat\W_{LD},\hat\W_{RD}\big)=\frac{1}{\rho_{\mathrm{ref}}}\,\mathrm{diag}\big(c_{\mathrm{ref}}^{-2},c_{\mathrm{ref}}^{-3},c_{\mathrm{ref}}^{-3},c_{\mathrm{ref}}^{-4}\big)\,\D.
\end{equation}
The normalized densities and pressures are not independent:
\begin{equation}
\sum_q\frac{\rho_q}{\rho_{\mathrm{ref}}}=4,\qquad\sum_q\frac{p_q}{p_{\mathrm{ref}}}=4 .
\end{equation}
Hence their dependence can be represented by three dimensionless contrasts each,
\begin{equation}
\boldsymbol\delta_\rho:=\frac14\,K\,\frac{\boldsymbol\rho}{\rho_{\mathrm{ref}}},\qquad
\boldsymbol\delta_p:=\frac14\,K\,\frac{\mathbf p}{p_{\mathrm{ref}}},\qquad
K:=\begin{pmatrix}1&-1&-1&\phantom{-}1\\ 1&\phantom{-}1&-1&-1\\ 1&-1&\phantom{-}1&-1\end{pmatrix},
\label{eq:contrasts}
\end{equation}
where $\boldsymbol\rho=(\rho_{RU},\rho_{LU},\rho_{LD},\rho_{RD})^{\mathsf T}$ and likewise $\mathbf p$.

This reduces the inputs to 14,
\begin{equation}
\boldsymbol\xi:=\Big(\boldsymbol\delta_\rho,\ \boldsymbol\delta_p,\ \frac{\mathbf u}{c_{\mathrm{ref}}},\ \frac{\mathbf v}{c_{\mathrm{ref}}}\Big)\in\mathbb R^{14},
\label{eq:inputs}
\end{equation}
where $\mathbf u$ and $\mathbf v$ collect the velocity components of the four states, the integrated flux difference admits the invariant representation
\begin{equation}
\D=\rho_{\mathrm{ref}}\,\mathrm{diag}\big(c_{\mathrm{ref}}^2,c_{\mathrm{ref}}^3,c_{\mathrm{ref}}^3,c_{\mathrm{ref}}^4\big)\,\hat\D(\boldsymbol\xi),\qquad\hat\D\in\mathbb R^4 .
\label{eq:Dscale}
\end{equation}
\subsection{Consistency and mirror symmetry}
\label{sec:net_consistency_symmetry}
If the data are one-dimensional across the face, $\W_{LU}=\W_{LD}$ and $\W_{RU}=\W_{RD}$, then $\D=\mathbf 0$. We impose this by a product: the multilayer perceptron (MLP) outputs a matrix $\mathrm{MLP}(\boldsymbol\xi)\in\mathbb R^{4\times8}$, which multiplies the jumps across the line $y=0$ on either side of the face,
\begin{equation}
\Delta\mathbf v:=\begin{pmatrix}\hat\W_{LU}-\hat\W_{LD}\\ \hat\W_{RU}-\hat\W_{RD}\end{pmatrix}\in\mathbb R^{8}.
\label{eq:Mdv}
\end{equation}
For data that are one-dimensional across the face, $\Delta\mathbf v=\mathbf 0$, so the product $\mathrm{MLP}(\boldsymbol\xi)\,\Delta\mathbf v =\mathbf 0$. 

The mirror $\sigma_x$, $x\to-x$, exchanges the left and right quadrants and reverses $u$, so $\D(\sigma_x\W)=S\,\D(\W)$ with $S=\mathrm{diag}(-1,1,-1,-1)$. Note that $\sigma_x^2=I$ and $S^2=I$. Mirror symmetry and consistency can then be enforced together via symmetrization,
\begin{equation}
\hat\D=\tfrac12\Big[\mathrm{MLP}(\boldsymbol\xi)\,\Delta\mathbf v+S\,\mathrm{MLP}(\boldsymbol\xi_\sigma)\,\Delta\mathbf v_\sigma\Big],
\label{eq:mirror}
\end{equation}
where $\boldsymbol\xi_\sigma$ and $\Delta\mathbf v_\sigma$ are the inputs and jumps of the mirrored problem $\sigma_x\W$.

For clarity, we restate the complete approximation:
{\setlength{\fboxsep}{7pt}%
\begin{empheq}[box=\fbox]{equation}
\begin{aligned}
&\D\approx\rho_{\mathrm{ref}}\,\mathrm{diag}\big(c_{\mathrm{ref}}^2,c_{\mathrm{ref}}^3,c_{\mathrm{ref}}^3,c_{\mathrm{ref}}^4\big)\,\hat\D,\\
&\hat\D=\tfrac12\big[\mathrm{MLP}(\boldsymbol\xi)\,\Delta\mathbf v+S\,\mathrm{MLP}(\boldsymbol\xi_\sigma)\,\Delta\mathbf v_\sigma\big],\qquad \mathrm{MLP}(\boldsymbol\xi)\in\mathbb R^{4\times8},\\[8pt]
&\rho_{\mathrm{ref}}=\tfrac14{\textstyle\sum_q}\rho_q,\qquad p_{\mathrm{ref}}=\tfrac14{\textstyle\sum_q}p_q,\qquad c_{\mathrm{ref}}=\sqrt{p_{\mathrm{ref}}/\rho_{\mathrm{ref}}},\\
&\hat\W_q=\big(\rho_q/\rho_{\mathrm{ref}},\ u_q/c_{\mathrm{ref}},\ v_q/c_{\mathrm{ref}},\ p_q/p_{\mathrm{ref}}\big),\\
&\boldsymbol\xi=\big(\tfrac14K\boldsymbol\rho/\rho_{\mathrm{ref}},\ \tfrac14K\mathbf p/p_{\mathrm{ref}},\ \mathbf u/c_{\mathrm{ref}},\ \mathbf v/c_{\mathrm{ref}}\big)\in\mathbb R^{14},\\
&\Delta\mathbf v=\big(\hat\W_{LU}-\hat\W_{LD},\ \hat\W_{RU}-\hat\W_{RD}\big)\in\mathbb R^{8},\\
&K=\left(\begin{smallmatrix}1&-1&-1&\phantom{-}1\\ 1&\phantom{-}1&-1&-1\\ 1&-1&\phantom{-}1&-1\end{smallmatrix}\right),\qquad S=\mathrm{diag}(-1,1,-1,-1),\\
&(\boldsymbol\xi_\sigma,\Delta\mathbf v_\sigma):\ \text{the mirrored problem's inputs and jumps}.
\end{aligned}
\label{eq:network}
\end{empheq}}%
The other three faces follow by rotation.


\section{Numerical results}
\label{sec:results}

\subsection{Training Settings}

\subsubsection{Training Data Generation}
\label{sec:data}
As in~\cite{zhang2026learning}, the training samples are not taken from simulations of the test problems. Two observations about simulations restrict the sampling to a relevant subset of all 2D Riemann problems:
\begin{enumerate}
  \item \textbf{Bounded variable range.} The states encountered in the target simulations lie within known intervals.
  \item \textbf{Similarity of the four states.} At most vertices of a simulation the four states are close, except near strong discontinuities.
\end{enumerate}
Based on these observations, we combine two sampling strategies. First, for global coverage, we draw the four states independently: $\rho$ and $p$ are log-uniform distributed on $[0.1,1.5]$ and $[0.08,2]$, respectively, and the velocity has a Mach number uniformly distributed on $[0,1.5]$ and a uniformly random direction. Second, for the many vertices with small jumps, we draw one state $(\rho_0,u_0,v_0,p_0)$, and perturb the sixteen numbers independently,
\begin{equation}
\begin{aligned}
\rho_q&=\rho_0\,(1+\varepsilon\,\delta_{q1}), &\qquad u_q&=u_0+\varepsilon\,c_0\,\delta_{q2},\\
p_q&=p_0\,(1+\varepsilon\,\delta_{q4}), & v_q&=v_0+\varepsilon\,c_0\,\delta_{q3},
\end{aligned}
\label{eq:perturbation}
\end{equation}
where $c_0=\sqrt{\gamma p_0/\rho_0}$ is the sound speed of the drawn state and the $\delta_{qi}$ are uniform in $[-1,1]$. The size $\varepsilon$ is log-uniform distributed on $[10^{-3},0.5]$, so small perturbations are more common.

Each problem is solved on $[-1,1]^2$ with $256^2$ cells by the unsplit MUSCL--Hancock scheme with the exact Riemann solver at CFL $0.4$, up to the final time
\begin{equation}
T = \frac{0.8}{\max(\lambda_c,\ s_{\max})},
\end{equation}
where $\lambda_c=\max_q\big(\sqrt{u_q^2+v_q^2}+c_q\big)$, and $s_{\max}$ is the largest absolute wave speed of the four one-dimensional Riemann problems. With this choice, the estimated fastest wave travels $80\%$ of the way across its quadrant by time $T$. $\D$ ~\eqref{eq:D} is computed from the four edges of the last snapshot. As a safeguard, we monitor the two outermost layers of cells along the edges. If their states differ meaningfully ($5\%$ of the largest initial jump) at any snapshot, $T$ is halved and the problem is solved again. Testing 2D Riemann problems on $1024^2$ cells changes $\D$ by $0.0$--$1.5\%$ across most problems. Problems whose solution would contain a vacuum are rejected. Across each face we use the vacuum condition of~\cite{toro2013riemann} (Eq.~(4.82)). In the simulation, we reject a problem if the minimum density or pressure in its snapshot falls below $10^{-3}$. No problems were rejected under our stated sampling approach. We note that vacuum states are not the focus of this work, and extending the framework to these regimes will be investigated in future work. We use 540,000 problems with close states and 60,000 problems with four independent states. Each problem yields four training samples, one per face.

\subsubsection{Neural Solver Model}
The MLP has four hidden layers of widths 32-64-64-64 with $ReLU$ activations. Since $\D$ vanishes with the jumps, the loss is a relative squared error,
\begin{equation}
\mathcal L=\frac1N\sum_{n=1}^{N}\frac{\big\|\hat\D_{\mathrm{pred},n}-\hat\D_{\mathrm{label},n}\big\|^2}{s_n^2+tol},\qquad tol=10^{-6},
\label{eq:loss}
\end{equation}
where $s_n$ is the jump size of vertex problem $n$, the root-sum-square deviation of its four normalized states~\eqref{eq:Whats} from their average,
\begin{equation}
s_n=\Big(\sum_q\big\|\hat\W_q-\overline{\hat\W}\big\|^2\Big)^{1/2},\qquad \overline{\hat\W}=\frac14\sum_q\hat\W_q .
\label{eq:jumpsize}
\end{equation}

We train with Adam \cite{kingma2017adammethodstochasticoptimization} and a cosine learning-rate schedule from $10^{-3}$ to $10^{-5}$ over $10\,000$ epochs, and keep the network with the lowest validation loss. Of the $2\,400\,000$ training samples, we use a 90/10 split between training and validation. The model is implemented in JAX \cite{jax2018github}.

For $\D$ itself, on the validation set, the mean relative error $\|\D_{\mathrm{pred}}-\D_{\mathrm{label}}\|/\|\D_{\mathrm{label}}\|$ is 0.7\% for problems with close states and 2.5\% for those with four independent states. We call the neural solver model SPIDERS (Structure-Preserving, Informed-by-Data Euler Riemann Solver)
\subsection{Numerical scheme}
\label{sec:scheme}

In the numerical tests, we compare two schemes, MH+ExactRS and MH+SPIDERS. They differ only in the face flux, and both use the unsplit MUSCL--Hancock method~\cite{toro2013riemann, van1984relation} with the monotonized central slope limiter. The two schemes share the Hancock predictor step, in which the reconstructed states are evolved by half a time step using the flux function evaluated at the reconstructed face midpoints. After that, MH+ExactRS uses the exact one-dimensional Riemann solver, computed from the two evolved states at the face midpoint. For MH+SPIDERS, we approximate the space--time flux integral by combining the one-dimensional face flux with the two vertex corrections. For example, for the face normal to the $x$-direction:
\begin{equation}
\mathbf I=\Delta y\,\Delta t\,\Fone+\frac{\Delta t^2}{2}\big(\D_{\mathrm{lower}}+\D_{\mathrm{upper}}\big),
\label{eq:spiders}
\end{equation}
where $\D_{\mathrm{lower}}$ and $\D_{\mathrm{upper}}$ are the integrated flux differences of the vertices at the lower and upper ends of the face, computed by SPIDERS from the four evolved corner states around each vertex. Here, $\Fone$ is the exact one-dimensional Riemann flux computed from the two evolved states at the face midpoint.

In both schemes the time step is
\begin{equation*}
\Delta t=\mathrm{CFL}\cdot\min\Big(\frac{\Delta x}{\max(|u|+c)},\ \frac{\Delta y}{\max(|v|+c)}\Big),
\end{equation*}
with the maxima taken over all cells and $c=\sqrt{\gamma p/\rho}$ the sound speed.
\subsection{Implosion}
\label{sec:implosion}
The computational domain of the implosion problem is the square $0\le x\le0.3$, $0\le y\le0.3$, with wall boundary conditions on all four sides. The initial condition consists of two constant states separated by a diagonal interface,
\begin{equation*}
(\rho,u,v,p)(x,y,0)=
\begin{cases}
(1.0,\,0,\,0,\,1.0), & x+y>0.15,\\[2pt]
(0.125,\,0,\,0,\,0.14), & x+y\le0.15,
\end{cases}
\end{equation*}
with $\gamma=1.4$, and the solution is computed up to $t=2.5$. The problem was popularized by Liska and Wendroff~\cite{liska2003comparison} to compare eight numerical schemes. It is challenging because it involves strong shock interactions, symmetry-driven wave focusing and the formation of a thin jet along the diagonal. The distance the jet travels is a measure of the numerical diffusion of a scheme~\cite{Skinner2019Fornax}. As in the original study, we use $400\times400$ cells. We compare MH+ExactRS at CFL~0.45 with MH+SPIDERS at CFL~0.9, so MH+SPIDERS needs half as many time steps (7329 against 14496).

Figure~\ref{fig:implosion} shows the pressure and the density contours at $t=2.5$ in the format of~\cite{liska2003comparison}. The two schemes produce the same pattern of reflected shocks and contact discontinuities, but the jet travels farther with MH+SPIDERS. Its head, the outermost vortex pair along the diagonal, reaches $x=y\approx0.158$ with MH+SPIDERS against $0.149$ with MH+ExactRS, about 6\% farther from the corner. Since numerical diffusion slows the jet~\cite{Skinner2019Fornax}, this indicates that MH+SPIDERS is less diffusive than MH+ExactRS.

\begin{figure}[!htbp]
  \centering
  \includegraphics[width=\linewidth]{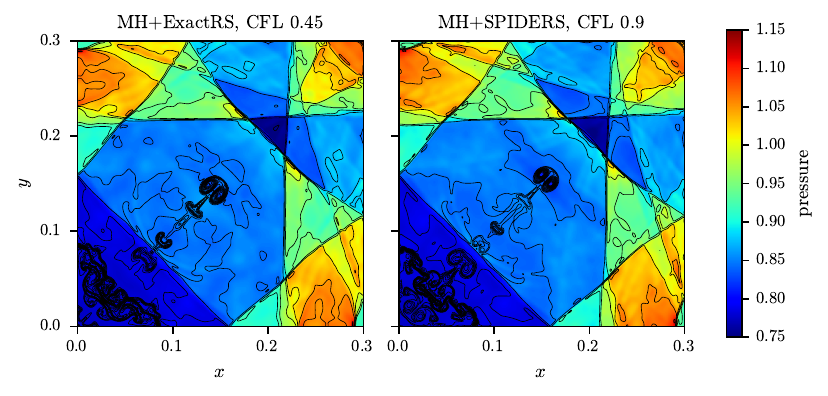}
  \caption{Implosion at $t=2.5$ on $400\times400$ cells: pressure (colour) and 31 density contours from 0.35 to 1.1. Left: MH+ExactRS at CFL~0.45. Right: MH+SPIDERS at CFL~0.9.}
  \label{fig:implosion}
\end{figure}

\subsection{Cylindrical explosion}
\label{sec:explosion}
The cylindrical explosion of Toro~\cite[Section~17.1]{toro2013riemann} has a radially symmetric solution, and it allows the error of a scheme to be measured against a one-dimensional reference solution. The computational domain is the square $0\le x\le2$, $0\le y\le2$, with transmissive boundaries. The gas is initially at rest, with
\begin{equation*}
(\rho,p)(x,y,0)=
\begin{cases}
(1.0,\,1.0), & (x-1)^2+(y-1)^2<0.4^2,\\[2pt]
(0.125,\,0.1), & \text{otherwise},
\end{cases}
\end{equation*}
and $\gamma=1.4$. As in~\cite{toro2013riemann}, cells cut by the circle receive the area-weighted average of the two conserved states, and we use $101\times101$ cells up to $t=0.25$. The reference solution is computed from the radially symmetric Euler equations, which carry the geometric source term $-(\rho u_r,\ \rho u_r^2,\ u_r(E+p))/r$, on $96\,000$ cells with the one-dimensional MUSCL--Hancock scheme. We compare MH+ExactRS at CFL~0.45 with MH+SPIDERS at CFL~0.9. MH+ExactRS fails at CFL~0.9 after five time steps. MH+SPIDERS needs 29 time steps against 58.

Figure~\ref{fig:explosion} shows the density at $t=0.25$. Both schemes capture the outgoing shock and contact discontinuity and the rarefaction that moves towards the centre. On this coarse grid, the density contours at the contact show an imprint of the grid in both schemes.

\begin{figure}[!htbp]
  \centering
  \includegraphics[width=\linewidth]{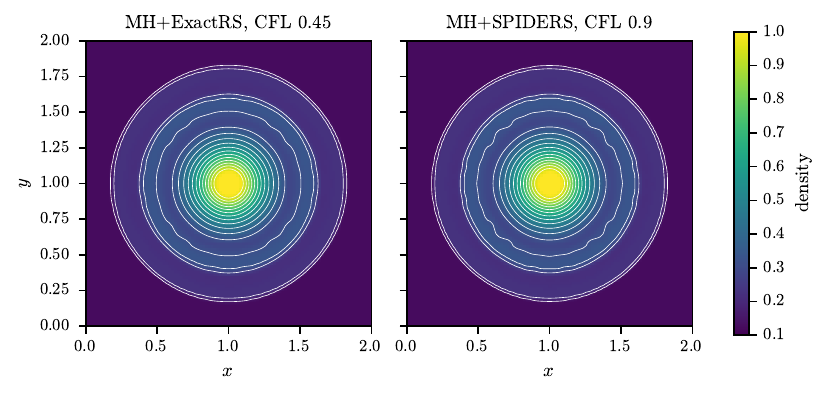}
  \caption{Cylindrical explosion at $t=0.25$ on $101\times101$ cells: density with 15 contours from 0.15 to 0.95. Left: MH+ExactRS at CFL~0.45. Right: MH+SPIDERS at CFL~0.9.}
  \label{fig:explosion}
\end{figure}

Figure~\ref{fig:explosion_radial} compares the solutions to the one-dimensional reference as functions of the distance $r$ from the center. The cells of MH+SPIDERS lie close to the reference (top row). Near the shock, the profile of MH+ExactRS is smeared over a slightly wider range of $r$ than that of MH+SPIDERS (mid row). The cumulative $L_1$ error is plotted against $r$ (bot row). It shows that the error is concentrated in the rarefaction, between its head ($r\approx0.1$) and its tail ($r\approx0.42$), at the contact ($r\approx0.62$), and at the shock ($r\approx0.82$). MH+SPIDERS is more accurate in the rarefaction and at the shock, and the two schemes are equally accurate at the contact. Over the entire domain, the $L_1$ errors of MH+SPIDERS are 9.0\%, 14.9\%, and 13.5\% lower than those of MH+ExactRS for the density, radial velocity, and pressure. A similar error difference is observed on $201\times201$, $401\times401$, and $801\times801$ cells (Table~\ref{tab:explosion}).

\begin{figure}[!htbp]
  \centering
  \includegraphics[width=\linewidth]{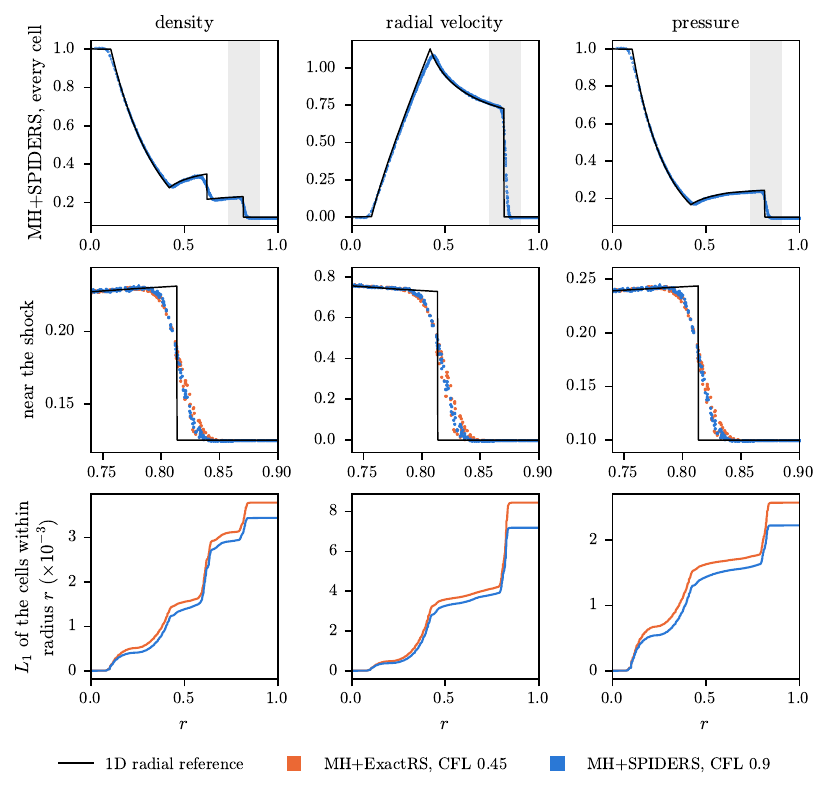}
  \caption{Cylindrical explosion at $t=0.25$ on $101\times101$ cells, against the distance $r$ from the centre, for MH+ExactRS at CFL~0.45 and MH+SPIDERS at CFL~0.9. Top: density, radial velocity and pressure of every cell of MH+SPIDERS and the one-dimensional radial reference. Middle: every cell of both schemes near the shock (shaded in the top row). Bottom: the cumulative $L_1$ error of the cells within radius $r$.}
  \label{fig:explosion_radial}
\end{figure}

\begin{table}[!htbp]
  \centering
  \caption{Cylindrical explosion at $t=0.25$: $L_1$ errors ($\times10^{-3}$) against the one-dimensional radial reference averaged over each cell, for MH+ExactRS at CFL~0.45 and MH+SPIDERS at CFL~0.9.}
  \label{tab:explosion}
  \setlength{\tabcolsep}{4.5pt}
  \begin{tabular}{rcccccc}
    \toprule
    & \multicolumn{2}{c}{density} & \multicolumn{2}{c}{radial velocity} & \multicolumn{2}{c}{pressure}\\
    \cmidrule(lr){2-3}\cmidrule(lr){4-5}\cmidrule(lr){6-7}
    Grid & ExactRS & SPIDERS & ExactRS & SPIDERS & ExactRS & SPIDERS\\
    \midrule
    $101^2$ & 3.79 & 3.44 & 8.45 & 7.19 & 2.57 & 2.22\\
    $201^2$ & 2.07 & 1.89 & 4.43 & 3.80 & 1.33 & 1.13\\
    $401^2$ & 1.15 & 1.05 & 2.21 & 1.91 & 0.67 & 0.57\\
    $801^2$ & 0.646 & 0.594 & 1.13 & 0.984 & 0.339 & 0.290\\
    \bottomrule
  \end{tabular}
\end{table}

\subsection{CFL number and computational cost}
\label{sec:cost}

Figure~\ref{fig:cfl} presents the density errors of both schemes for CFL numbers ranging from 0.4 to 1.1 in increments of 0.05, for the implosion and explosion tests. Because the implosion has no exact solution, its error is measured against a run of MH+ExactRS at CFL~0.45 on $800\times800$ cells. For the explosion, we use the one-dimensional radial reference described in Section~\ref{sec:explosion}. MH+ExactRS fails at CFL~0.75 in the implosion test and at CFL~0.8 in the explosion test. The implosion error of MH+ExactRS begins to increase rapidly near CFL~0.55, whereas the explosion error begins to increase rapidly at approximately CFL~0.75.

\begin{figure}[!htbp]
  \centering
  \includegraphics[width=\linewidth]{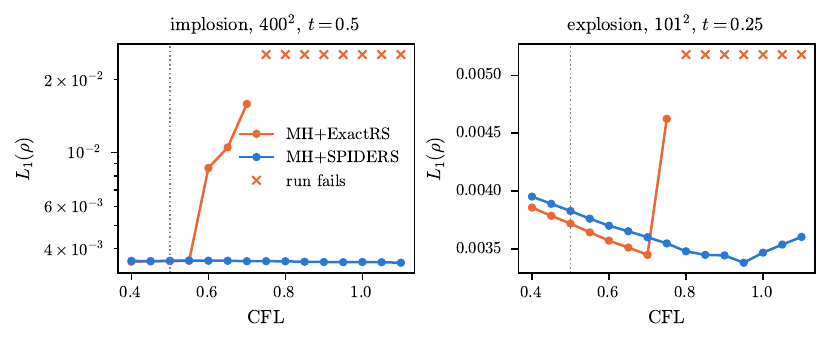}
  \caption{Density error against the CFL number for MH+ExactRS and MH+SPIDERS: the implosion at $t=0.5$ on $400\times400$ cells (left) and the explosion at $t=0.25$ on $101\times101$ cells (right). Crosses mark runs that fail. The dotted line marks CFL~0.5.}
  \label{fig:cfl}
\end{figure}

Table~\ref{tab:cost} reports the wall-clock times of both schemes on a single NVIDIA GH200 GPU. MH+SPIDERS is faster on every grid. It requires half as many time steps as MH+ExactRS, although each time step is more computationally expensive because vertex corrections are computed in addition. On the coarsest grids, the GPU is not fully utilized, and MH+SPIDERS requires about 40\% less time. When the GPU is fully utilized, MH+SPIDERS requires about 30\% less time.  

\begin{table}[!htbp]
  \centering
  \caption{Wall times in seconds of MH+ExactRS at CFL~0.45 and MH+SPIDERS at CFL~0.9 on one NVIDIA GH200 GPU, for the implosion up to $t=2.5$ and the explosion up to $t=0.25$. The time of MH+SPIDERS relative to MH+ExactRS is given in brackets.}
  \label{tab:cost}
  \setlength{\tabcolsep}{4.5pt}
  \begin{tabular}{lcccc}
    \toprule
    Implosion & $100^2$ & $200^2$ & $400^2$ & $800^2$\\
    \midrule
    MH+ExactRS & 1.35 & 3.11 & 12.6 & 122\\
    MH+SPIDERS & 0.755 (56\%) & 1.91 (61\%) & 9.24 (73\%) & 85.7 (70\%)\\
    \midrule
    Explosion & $101^2$ & $201^2$ & $401^2$ & $801^2$\\
    \midrule
    MH+ExactRS & 0.0257 & 0.0565 & 0.229 & 1.77\\
    MH+SPIDERS & 0.0148 (58\%) & 0.0355 (63\%) & 0.169 (74\%) & 1.29 (73\%)\\
    \bottomrule
  \end{tabular}
\end{table}

\section{Conclusions}
\label{sec:conclusions}
We have shown that, in any number of space dimensions and for CFL numbers up to one, the exact space--time flux integral across the face of a multidimensional Riemann problem is a polynomial in the time step and the face dimensions. Learning a multidimensional Riemann solver therefore reduces to learning a map from the initial states to the coefficients of this polynomial. Building on this result, we proposed a structure-preserving, data-informed approach to multidimensional Riemann solvers, in which a neural network is trained on high-resolution simulations of the Riemann problem computed once and offline. The neural Riemann solver enforces crucial properties, such as consistency and mirror symmetry. Unlike many classical multidimensional Riemann solvers, whose accuracy is limited by simplified models, the neural framework approaches the exact flux integral, provided that the high-resolution simulations converge to it and that the network predicts the coefficients accurately.

We demonstrated this approach in an unsplit MUSCL--Hancock scheme for the two-dimensional Euler equations. The neural Riemann solver, SPIDERS (Structure-Preserving, Informed-by-Data Euler Riemann Solver), allows a CFL number of 0.9, at which the same scheme with the exact one-dimensional Riemann solver fails. Compared with that scheme at a CFL number of 0.45, MH+SPIDERS needs half as many time steps and is more accurate. In the implosion problem, its jet travels about 6\% farther, indicating less numerical diffusion. In the cylindrical explosion, it preserves the problem's radial symmetry as well as MH+ExactRS, its $L_1$ errors are 8--15\% lower on all grids from $101\times101$ to $801\times801$ cells, and it converges at the same rate as MH+ExactRS. MH+SPIDERS takes about 30\% less computational time than MH+ExactRS on the finer grids.

Several limitations remain. A network is reliable only for states similar to those it was trained on: the structural properties of SPIDERS hold for every input, but its accuracy on states outside the training distribution is not guaranteed. The training data should therefore cover the range of states that a simulation produces, for example by estimating this range in advance or by generating data adaptively from states encountered in simulations. We have not considered states near vacuum. We have demonstrated the framework only for the two-dimensional Euler equations and compared it only with the same scheme using the exact one-dimensional Riemann solver, not with other multidimensional Riemann solvers. 

Future work will apply the framework in three spatial dimensions and to more complex systems, such as magnetohydrodynamics, address states near vacuum, and compare the resulting solvers with existing multidimensional Riemann solvers.

\section*{CRediT Authorship Contribution Statement}
\textbf{Yucheng Zhang:} Conceptualization, Methodology, Software, Validation, Formal analysis, Writing – original draft, Writing – review \& editing, Visualization. 

\textbf{Clint Dawson:} Conceptualization, Methodology, Writing – review \& editing, Supervision, Project administration, Funding acquisition.

\section*{Declaration of Competing Interest}
The authors declare that they have no known competing financial interests or personal relationships that could have appeared to influence the work reported in this paper.

\section*{Acknowledgments}
This work was supported by the U.S. National Science Foundation under Award No. 2208461.

The authors acknowledge the Texas Advanced Computing Center (TACC) at The University of Texas at Austin for providing computational resources that have contributed to the research results reported within this paper. URL: http://www.tacc.utexas.edu

\section*{Data Availability}
Data will be made available on request.



\bibliographystyle{elsarticle-num}
\bibliography{references}

\end{document}